\documentclass{article}

\usepackage{xcolor}
\usepackage{xargs} 
\usepackage[textsize=footnotesize]{todonotes}
\newcommandx{\td}[2][1=]{\todo[linecolor=green,
			backgroundcolor=green!10,bordercolor=green,#1]{TODO: #2}}
\newcommandx{\ms}[2][1=]{\todo[linecolor=blue,
			backgroundcolor=blue!10,bordercolor=blue,#1]{MS: #2}}
\newcommandx{\jy}[2][1=]{\todo[linecolor=red,
			backgroundcolor=red!10,bordercolor=red,#1]{JY: #2}}

\usepackage{amsfonts}       

\usepackage{amsthm}
\usepackage{mathtools}      
\usepackage{amssymb}        
\usepackage{filecontents}
\usepackage{graphicx}       
\usepackage{marginnote}     
\usepackage{marvosym}       
\usepackage{overpic}        
\usepackage{tabularx}
\usepackage{cite}
\usepackage{color}
\usepackage[linesnumbered,algoruled,boxed,lined]{algorithm2e}
\usepackage[normalem]{ulem}
\usepackage{epstopdf}
\usepackage{enumitem}
\usepackage[normalem]{ulem}
\usepackage{xspace}

\newtheorem{proposition}{Proposition}[section]
\newtheorem{lemma}{Lemma}[section]

\newtheorem{corollary}{Corollary}[section]
\newtheorem{theorem}{Theorem}[section]
\theoremstyle{definition}

\theoremstyle{remark}

\newif\ifarxiv
\arxivfalse

\usepackage{mathtools,xparse}

\makeatletter
\newcommand{\customlabel}[2]{%
\protected@write \@auxout {}{\string \newlabel {#1}{{#2}{}}}}
\makeatother

\newif\ifdraft
\draftfalse

\newcommand{\stps}{\texttt{STP}s\xspace}
\newcommand{\stp}{\texttt{STP}\xspace}
\newcommand{\st}{\texttt{ST}\xspace}

\newcommand{\bsn}{\texttt{BSN}\xspace}
\newcommand{\bsns}{\texttt{BSN}s\xspace}

\begin{document}
\date{}
\font\titlefont=ptmr at 16pt
\title{Budgeted Steiner Networks: Three Terminals with 
Equal Path Weights}
\author{Mario Szegedy \and Jingjin Yu 
}
\maketitle


%
\begin{abstract}
Given a set of terminals in 2D/3D, the network with the shortest 
total length that connects all terminals is a Steiner tree. On 
the other hand, with enough budget, every terminal can be connected 
to every other terminals via a straight edge, yielding a complete 
graph over all terminals. 
In this work, we study a generalization of Steiner trees asking
what happens in between these two extremes.
Focusing on three terminals with equal pairwise path weights, we 
characterize the full evolutionary pathway between the Steiner 
tree and the complete graph, which contains intriguing 
intermediate structures. 
\end{abstract}

\section{Introduction}\label{sec:intro}
Consider a scenario in which three or more terminals (e.g., the black nodes $A, B,$ and $C$ in Fig.~\ref{fig:ex}) are to be connected using a (graph) network, the total length of which is limited. 
%
%
\begin{figure}[htp]
\begin{center}
\begin{overpic}[width=4.2in,tics=5]
{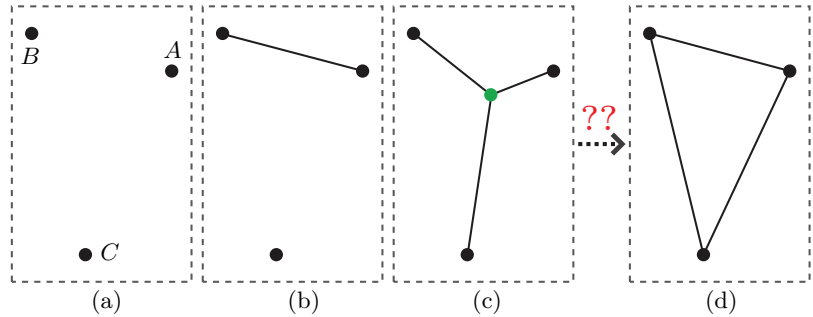}
\put(19,28.2){{\small $A$}}
\put(1.2,27.5){{\small $B$}}
\put(11.2,3){{\small $C$}}
\put(10,-3){{\small (a)}}
\put(34.5,-3){{\small (b)}}
\put(57.5,-3){{\small (c)}}
\put(86.5,-3){{\small (d)}}
\end{overpic}
\end{center}
\caption{Evolution of a \emph{budgeted Steiner network} over three (black) 
terminals as the budget increases. (a) Three terminals, $A, B$, and $C$, to be 
connected. (b) The minimal non-trivial network that connects two terminals. 
(c) The minimal network connecting all terminals, which is a Steiner tree. (d) With sufficient budget, the network is a complete graph. What happens between 
(c) and (d)?} 
\label{fig:ex}
\end{figure} 

At one extreme, the minimum length budget required to connect all terminals 
corresponds to the total length of the edges of a Steiner tree 
over the terminals (Fig.~\ref{fig:ex}(c)). 
The well-known \emph{Steiner tree problem} (\stp) seeks optimal network structures for 
connecting a set of terminals while minimizing the total edge lengths
\cite{winter1987steiner,hwang1992steiner}. \stp generally asks for a minimally 
connected network, resulting in a topology that is a tree.
At the other extreme, when there is no limit on the budget, the best network 
structure is clearly a complete graph over all terminals, where every pair of 
terminals are connected through a straight edge. Such a network ensures the 
shortest possible travel distance between any pair of terminals. 
What if, however, the budget falls between the two extremes? 

To address the question, we propose the \emph{budgeted Steiner network} (\bsn) 
problem/model.
As a natural generalization of \stp, \bsn seeks the best network structure 
for a given length budget to connect three or more terminals, which  
reside in $\mathbb R^d$ for some $d \ge 1$, such that the sum of the 
(weighted) distances between pairs of nodes are minimized. In this work, 
we mainly focus on the case of three terminals with $d = 2$ (for three 
terminals, $d = 2$ is the same as $d\ge 2$). 

The generalization immediately leads to rich and interesting structures, 
even when only three terminals are involved. 
As the budget increases, the network structure changes continuously 
between a Steiner tree and a complete graph over the terminals,
a few snapshots of which are illustrated in Fig.~\ref{fig:spectrum}. 
\begin{figure}[h]
\begin{center}
\begin{overpic}[width=4.5in,tics=5]
{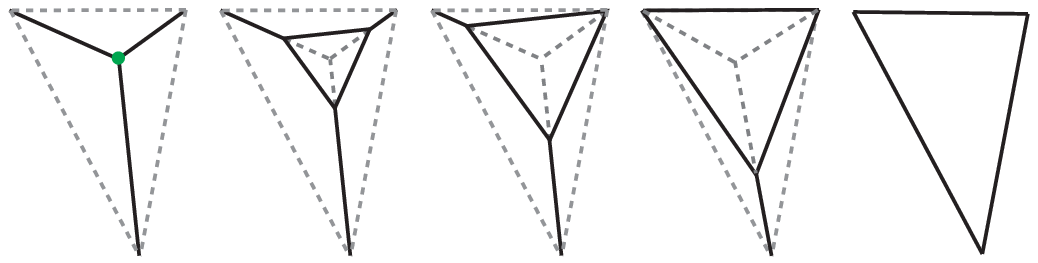}
\put(17.5, 25){{\small $A$}}
\put(-0.5,25){{\small $B$}}
\put(15.5,0){{\small $C$}}
\put(34.5, 23.5){{\small $A'$}}
\put(25,18){{\small $B'$}}
\put(33.5,14){{\small $C'$}}
\put(18,10){{$\to$}}
\put(39,10){{$\to$}}
\put(60,10){{$\to$}}
\put(81,10){{$\to$}}
\put(12,-4){{\small (a)}}
\put(32,-4){{\small (b)}}
\put(52,-4){{\small (c)}}
\put(72,-4){{\small (d)}}
\put(92,-4){{\small (e)}}
\end{overpic}
\end{center}
\caption{A spectrum of optimal Euclidean \bsn network structures (solid lines) 
for three terminals in a typical setup, as the allowed budget increases.}
\label{fig:spectrum}
\end{figure} 

As a summary of the full evolutionary pathway, if all internal angles of a 
$\triangle ABC$ are smaller than $2\pi/3$, the Steiner tree over terminals 
$A, B$, and $C$ has a Steiner point that is internal to the 
triangle (e.g., the green dot in 
Fig.~\ref{fig:spectrum}). 
In this case, for a generic $\triangle ABC$ (that is, $\triangle ABC$ is 
not an isosceles triangle), as the budget increases past the length of 
the Steiner tree, an equilateral triangle $\triangle A'B'C'$ will ``grow'' 
out of the Steiner point (Fig.~\ref{fig:spectrum}(b)) and continues to 
expand until one vertex of $\triangle A'B'C'$ meets one of the terminals,
say $A$.  
Past this point, $\triangle A'B'C'$ continues to expand as an isosceles 
triangle with $A' = A$ fixed (Fig.~\ref{fig:spectrum}(c)) as the budget 
continues to increase, until another vertex meets $B$ or $C$, say $B$. 
$\triangle A'B'C'$ then continue to expand with $A'=A$ and $B'=B$ fixed,
and $C'$ moving toward $C$, until it fully coincides with $\triangle ABC$. 
If $\triangle ABC$ has one angle equal to or larger than $2\pi/3$, the 
evolutionary pathway is similar but shortened; the corresponding \bsn 
does not have an initial phase containing an equilateral triangle. 

The main contribution of this work is the rigorous characterization 
of the precise evolution pathway of a \bsn as the available budget 
increases, for three arbitrarily located terminals. The analysis also 
implies an efficient algorithm for computing the optimal \bsn structure 
for any given budget.  
%
%
%

As a combination of multi-terminal shortest path problems 
and constrained structural optimization problems, \bsns have many 
potential real-world applications.
As an example, consider a scenario in which several islands that 
are only connected through ships or airplanes. To facilitate commerce 
activities among the islands, road bridges are to be built with a 
given bridge length budget.
The minimum budget for connecting all islands will yield a Steiner
tree. Below this budget, not all islands are connected. As budget 
increases, travel time between islands can be improved. For a given 
budget, having algorithms for \bsn can help design the best network 
structure, maximizing the utility of the fixed budget. 
On the other hand, it is rarely the case that sufficient budget will 
be provided so that the bridges will form a complete graph. 
Similar scenarios appear in constructing intra- and inter-city highways, 
balancing warehouse space allocations for storage and transportation, 
and designing walkways connecting buildings on a college campus or road
networks for moving components using autonomous vehicles in a large factory.

\section{Related Work}\label{sec:related} 
\bsn problems are closely related to \stps 
\cite{winter1987steiner,hwang1992steiner,hauptmann2013compendium}, 
which is a broad term covering a class of network optimization problems. 
%
An \stp seeks a minimal network that connects a set of terminals (in 
Euclidean space or on graphs that are possibly edge/vertex weighted). 
There are four main cases: Euclidean, rectilinear, discrete/graph-theoretic
\cite{cook1971complexity,karp1972reducibility}, 
and phylogenetic \cite{hwang1992steiner}. Considering the paper's scope, 
we provide a brief literature review of Euclidean \stps. 

The Euclidean \stp asks the following 
question: given $n$ terminals in 2D or 3D, find a network that connects
all $n$ points with the minimum total length (the discussion from now 
on will be limited to the 2D case). Obviously, the resulting network is a 
tree and may only have straight line segments; it may also require 
additional intermediary nodes to be added. These added nodes are called 
{\em Steiner points}. 
The study of Euclidean \stp bears with it a long history; the initial 
mathematical study of the subject may be traced back to at least 1811 
\cite{brazil2014history}. 
According to \cite{korte2001vojtvech}, key properties of Euclidean \stp 
have been established in (as early as) the 1930s by Jarn{\'\i}k 
and K{\"o}ssler \cite{jarnik1934minimalnich}. An interconnecting network 
$T$ is called a Steiner tree if it satisfies the following conditions 
\cite{hwang1992steiner}:
\begin{itemize}
\item[{(a)}] $T$ is a tree,
\item[{(b)}] Any two edges of $T$ meet at an angle of at least $2\pi/3$, and
\item[{(c)}] Any Steiner point cannot be of degree 1 or 2. 
\end{itemize}

These conditions turn out to be also relevant in our study of the \bsn problem.
The solution to an Euclidean \stp must be a Steiner tree. Note that {(b)} implies 
a node of the network has a maximum degree of 3. 
Together, {(b)} and {(c)} imply that three edges must meet at a Steiner point 
forming angles of $2\pi/3$ in a pairwise manner (see Fig.~\ref{fig:ex} and 
Fig.~\ref{fig:st-ex}(a)). 
Because Euclidean Steiner trees assume minimal energy configurations, they also 
appear in nature (Fig.~\ref{fig:st-ex}(b)); in fact, it is possible to employ 
related natural phenomena (e.g., using rubber bands and soap film) to ``compute'' 
Euclidean Steiner trees \cite{miehle1958link,clark1981communication,gilbert1968steiner}.
%

\begin{figure}[htp]
\begin{center}
\vspace*{3mm}
\begin{overpic}[width=3.4in,tics=5]
{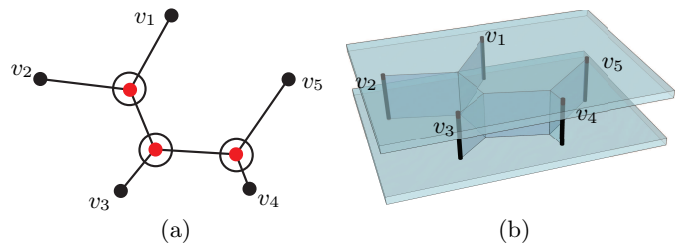}
\put(16,30){{\small $v_1$}}
\put(-3,22){{\small $v_2$}}
\put(9,1.5){{\small $v_3$}}
\put(35,2){{\small $v_4$}}
\put(41.8,20){{\small $v_5$}}
\put(70,27){{\small $v_1$}}
\put(50,20){{\small $v_2$}}
\put(62,13){{\small $v_3$}}
\put(84,15){{\small $v_4$}}
\put(88,23){{\small $v_5$}}
\put(20,-3){{\small (a)}}
\put(72,-3){{\small (b)}}
\end{overpic}
\end{center}
\caption{An Euclidean Steiner tree and a real-world realization. (a) A 
Steiner tree for five fixed terminals $v_1$--$v_5$ (black nodes). The 
solution contains three Steiner points (red nodes). The lines meet at 
these points form $2\pi/3$ angles. (b) The illustration of a physical 
experiment. If 5 pegs are sandwiched between two glass panels and the 
setup is immersed into and then lifted from soapy water, the soap film 
will naturally form a Steiner tree that matches the one in (a).}
\label{fig:st-ex}
\end{figure} 

Computing an Euclidean \stp is NP-hard, although there is a polynomial time 
approximation scheme (PTAS) for solving it \cite{Aro98}. On the more practical 
side, fast methods including the GeoSteiner algorithm \cite{winter1997euclidean,
warme2000exact} have been developed building on the Melzak construction 
\cite{melzak1961problem}. An open source implementation of the GeoSteiner 
algorithm is maintained by the authors \cite{geoSteiner}.

%

\section{Preliminaries}\label{sec:problem} 
Let there be $n \ge 3$ \emph{terminals} $N = \{v_1, \ldots, v_n\}$, 
distributed in some way in a $d$-dimensional unit cube, $d > 0$. 
For each pair of terminals $v_i$ and $v_j$, $1 \le i < j \le n$, let 
$w_{ij} \in (0, 1]$ denote the (relative) \emph{weight} or \emph{importance} 
of the route connecting $v_i$ to $v_j$. In practice, $w_{ij}$ may model the 
expected traffic flow from $v_i$ to $v_j$, for example. 
In an Euclidean \emph{budgeted Steiner tree} (\bsn) problem, straight line 
segments are to be added for connecting the $n$ terminals so that some or 
all of the terminals are connected. 
Similar to Steiner trees, intermediate nodes other than $v_1, \ldots, 
v_n$, which we call \emph{anchors}, may be added. 
The terminals, anchors, and the straight line segments then form a graph
containing one or more connected components. 
Under the constraint that the total length of the line segments does not 
exceed a {\em budget} $L$, the \bsn problem seeks a network structure 
that minimizes the objective
\begin{align}\label{eq:obj}
J(L) = \sum_{1 \le i < j \le n} w_{ij}d_{ij},
\end{align}
in which $d_{ij}$ denotes the shortest distance between $v_i$ and $v_j$ 
on the network. If no path exists between $v_i$ and $v_j$, let $d_{ij}$ 
be some very large number. 

In the current work, we examine the case of $n = 3$ and $w_{ij} = 1$ for 
all $1 \le i,j \le 3$, $i \ne j$, i.e., paths between pairs terminals are equally important. 
Let the three terminals be $A, B$, and $C$, we are looking for a \bsn 
minimizing the sum $d_{AB} + d_{BC} + d_{AC}$ subject to the budget $L$.
%
%
For a fixed $L$, let $N(L)$ denote the optimal \bsn structure.  
Let $L_{\st}$ be the budget $L$ when $N(L)$ is a Steiner tree. 
For convenience, let $N_{\st} := N(L_{\st})$.

\section{Anchor Structures and Steiner Triangles}\label{sec:anchors}
\subsection{Basic Properties of Anchors}
We begin with analyzing what happens when $L = L_{\st} + \varepsilon$ 
for small $\varepsilon > 0$, for the case where the Steiner point lies inside 
$\triangle ABC$, which happens when all angles of $\triangle ABC$ are smaller 
than $2\pi/3$. 
Due to continuity, the resulting structure that minimizes 
Eq.~\eqref{eq:obj} must be a perturbation of $N_{\st}$ (e.g., 
Fig.~\ref{fig:ex}(b)). 
This means that $N(L_{\st} + \varepsilon)$ must start ``growing'' 
at the Steiner point. 
We want to understand how $N(L_{\st} + \varepsilon)$ evolves for 
small $\varepsilon$. 
This raises the following questions: (1) how many line segments are in 
$N(L = L_{\st} + \varepsilon)$ and (2) how do they come together? 
We note that $N(L_{\st} + \varepsilon)$ must contain more than 
three straight line segments. Otherwise, $N(L_{\st} + \varepsilon)$
will still be a tree but with $d_{AB} + d_{BC} + d_{AC} = 
2(L_{\st} + \varepsilon) > 2L_{\st}$, i.e., $J(L) > J(L_{\st})$.
%
%
%

To answer above-mentioned questions, we start with establishing essential 
properties of anchors, concerning their locations, degrees, and numbers. 
It is clear that anchors must always fall within $\triangle ABC$;
otherwise, an outside anchor (on the convex hull of all terminals 
and anchors) can be ``retracted'' toward the boundary of $\triangle 
ABC$ to reduce both the budget and the objective function value. 
In fact, anchors cannot reside on the boundary of $\triangle ABC$, 
as shown in the following lemma. 

%

\begin{lemma}[Interiority of Anchors]\label{l:anchor-int}
For three terminals $A, B$, and $C$, any anchor must fall in the interior 
of $\triangle ABC$, excluding its perimeter. \end{lemma}
\begin{proof}
Consider the setting illustrated in Fig.~\ref{fig:ai} where only a portion 
of $\triangle ABC$ is drwan. $A$ and $C$ are terminals. Suppose that $D$ is an 
anchor on $AC$ and the horizontal line segment passing through $D$ and $D'$ 
is part of an optimal network structure. For the setup, $DD'$ must be 
part of the shortest path on the optimal network that connects $A$ to $B$ 
as well as $C$ to $B$; the entire $AC$ must also be part of the network that
connects $A$ and $C$. 
\begin{figure}[h]
\begin{center}
\vspace{1mm}
\begin{overpic}[width=1.59in,tics=5]
{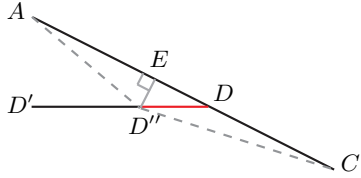}
\put(-8,50){{\small $A$}}
\put(102,0){{\small $C$}}
\put(60,23){{\small $D$}}
\put(32,12){{\small $D''$}}
\put(-8,19){{\small $D'$}}
\put(39,34){{\small $E$}}
\end{overpic}
\end{center}
\caption{Moving $C'$ along $C'C$ for a small amount.}
\label{fig:ai}
\end{figure} 

We claim that such a configuration cannot be optimal. To see this, 
retract $D$ along $DD'$ by some small distance of $|DD''|$. This reduces 
the budget by 
$$
\Delta L = |DD''| + (|AC| - |AD''| - |CD''|).
$$

At the same time, the cost reduction is 

$$
\Delta J = 2|DD''| + (|AC| - |AD''| - |CD''|).
$$ 

Let $E \in AC$ be a point such that $D''E \perp AC$. It is 
straightforward to derive that $|ED''| \gg |CD''| - |CE|$ and 
$|ED''| \gg |AD''| - |AE|$ for sufficiently small $|ED''| > 0$. 
Therefore, $|DD''| \ge |ED''| > (|AD''| + |CD''| - |AC|)$. 
This means that for small $|DD''|$, both $\Delta L$ and $\Delta J$ are
positive, i.e., we can reduce budget and at the same time reduce the 
cost by retracting $D$ along $DD'$ to $D''$. This means that $D$ cannot 
be an anchor. 
\end{proof}

Similar to Steiner points, each anchor must have degree exactly three.

\begin{lemma}[Degree of Anchors]\label{l:anchor-degree}
For three terminals, any anchor must have degree three. 
\end{lemma}
\begin{proof}
Each anchor must connect at least three line segments; otherwise, the 
anchor point and the involved line segments only cause increases to the 
objective $d_{AB} + d_{BC} + d_{AC}$. An anchor's degree also cannot 
be four or larger when there are only three terminals, because each outgoing 
edge from an anchor must be on a shortest path to a unique terminal, if we 
are to minimize Eq.~\eqref{eq:obj}. But there are only three terminals.
\end{proof}

In general, anchors have degree three even when there are more 
than three terminals. Building on Lemmas~\ref{l:anchor-int}
and~\ref{l:anchor-degree}, we continue to show that there can be at 
most three anchors for three terminals. 

\begin{lemma}[Number of Anchors in $N(L_{\st} + \varepsilon)$]
\label{l:anchor-number-3-isp}
When all angles of $\triangle ABC$ are below $2\pi/3$, for small $\varepsilon 
> 0$, $N(L_{\st} + \varepsilon)$ contains three anchors that forms a triangle 
inside $\triangle ABC$. 
\end{lemma}
\begin{proof}
By Lemma~\ref{l:anchor-degree}, all anchors have degree three. 
If there is only a single anchor that is not the Steiner point, then 
$N(L_{\st} + \varepsilon)$ still has a tree structure. This tree is different 
from $N_{\st}$ which is minimal, so the new tree must have a larger objective 
function value which cannot be optimal. 

If there are two anchors, each with degree three, then both of them 
cannot be connected to all of $A, B$, and $C$; there must be 
exactly five line segments in $N(L_{\st} + \varepsilon)$, one of which 
connecting the two anchors. This leaves four line segments connected 
to the three terminals, which means that two of these line segments
must reach the same terminal. This will induce a total budget that 
cannot be an arbitrarily small amount above $L_{\st}$ when the Steiner
point is inside $\triangle ABC$. That is, this is impossible with 
a budget $L_{\st} + \varepsilon$ for small $\varepsilon > 0$.

There cannot be more than three anchors when there are only three terminals. 
To establish this, we note that a shortest path between any two terminals, 
when there are three terminals in total, can make at most two ``turns'' due 
to path sharing. To see this, consider the shortest path $P_{AB}$ between 
terminals $A$ and $B$. $P_{AB}$ may bend at most two times, once to share with 
a path from $A$ to $C$ and once to share with a path from $B$ to $C$. 
If $P_{AB}$ bends once, say at an anchor $A'$, then both $AA'$ or $A'B$ must 
be on a shortest path to $B$ and we must have a tree. This is not possible 
under the assumption that $\varepsilon$ is small, so there can only be one 
edge coming out of a terminal. Therefore, each shortest path between two 
terminals must bend exactly twice at two anchors. 
The thee shortest paths then have a total of six anchors. 
Because each anchor is shared by two shortest paths, there can only be 
three unique anchors that form a triangle. 
\end{proof}

\subsection{Steiner Triangle for Three Anchors}
Having shown that there are three anchors, let the anchor closest
to $A, B$ and $C$ be $A', B'$, and $C'$, respectively. This suggest that 
$N(L_{\st} + \varepsilon)$ contains six line segments $AA'$, $BB'$, $CC'$, 
$A'B'$, $A'C'$, and $B'C'$. 
We call $\triangle A'B'C'$ that ``grows'' out of the Steiner point a 
\emph{Steiner triangle}. 
Next, we establish that $\triangle A'B'C'$ is an equilateral triangle, 
starting with showing that its three internal angles are bisected by 
$AA'$, $BB'$ and $CC'$. 
The objective Eq.~\eqref{eq:obj}, $d_{AB} + d_{BC} + d_{AC}$ for the current
setting, translates to 
\begin{align}\label{ex:obj2}
J(L_{\st} + \varepsilon) = 2|AA'| + 2|BB'| + 2|CC'| + |A'B'| + |A'C'| + |B'C'|.
\end{align}

\begin{lemma}[Bisector of Steiner Triangle]
\label{l:bisector}
For three terminals $A, B$, and $C$ with a Steiner point, let 
$N(L_{\st} + \varepsilon)$ be composed of the Steiner triangle 
$\triangle A'B'C'$ and segments $AA'$, $BB'$ and $CC'$. Then an angle of 
$\triangle A'B'C'$ is bisected by the line passing the corresponding anchor 
and the terminal the anchor is connected to. 
\end{lemma}
\begin{proof}
Assume that for a given budget $L = L_{\st} + \varepsilon$, the optimal network 
$N(L_{\st} + \varepsilon)$ has corresponding optimal objective 
$J({L_{\st} + \varepsilon})$ as given in Eq.~\ref{ex:obj2}. 
We show that $CC'$ is a bisector of $\angle A'C'B'$ by analyzing the local 
changes to $L$ and $J({L_{\st} + \varepsilon})$ if we perturb $C'$. 
\vspace*{2mm}
\begin{figure}[h!]
\begin{center}
\begin{overpic}[width=1.84in,tics=5]
{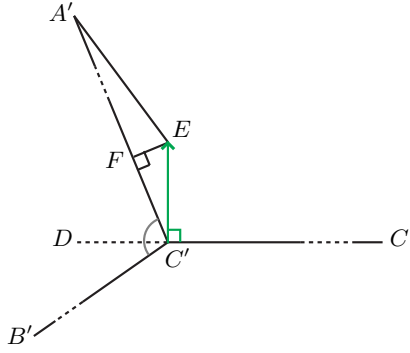}
\put(102,26){{\small $C$}}
\put(5,90){{\small $A'$}}
\put(-7,-2){{\small $B'$}}
\put(38,20){{\small $C'$}}
\put(6,26){{\small $D$}}
\put(40,57){{\small $E$}}
\put(21,48){{\small $F$}}
\end{overpic}
\end{center}
\caption{Perturbing $C'$ in an assumed optimal configuration for the 
three-terminal Euclidean \bsn problem. The figure zooms in around $C'$ without
showing $A$ and $B$. The drawing intentionally avoids assuming that 
$\triangle A'B'C'$ is an equilateral triangle.}
\label{fig:bisector}
\vspace*{-2mm}
\end{figure} 

Referring to Fig.~\ref{fig:bisector}, let $D$ be a point on the extension
of $\overrightarrow{CC'}$. A point $E$ is introduced that shifts $C'$ up
vertically (i.e., $C'E \perp C'C$) by the amount $|C'E|$, as a small 
perturbation to $C'$. Now draw a line $EF$ such that $EF \perp A'C'$ with 
$F \in A'C'$. Because $|C'E|$ is small, $|A'F| \approx |A'E|$ (this is a 
second order approximation). As $C'$ is moved to $E$, the length change 
of $A'C'$ is given by $|A'E| - |A'C'|$, which is approximately 
\[
|A'F| - |A'C'| = -|FC'| = -|C'E|\cos \angle A'C'E = -|C'E|\sin \angle A'C'D. 
\]

Following a similar analysis procedure, the length change of $B'C'$, $|B'E|-
|B'C'|$, is approximately $|C'E|\sin \angle B'C'D$. Because $C'E \perp C'C$ 
and $|C'E|$ is small, $|CC'|\approx |CE|$ (also a second order approximation).
Relating the length changes due to moving $C'$ up 
to the change of the budget $L$, the net change to $L$ is $|C'E|(\sin 
\angle B'C'D - \sin \angle A'C'D)$ (i.e., $B'C'$ becomes longer and $A'C'$ 
becomes shorter with $CC'$ unchanged, as a second order approximation). 
The change to the objective $J(L_{\st} + \varepsilon)$ is the same since 
$CC'$ is unaffected by $C'E$. 

Because the changes to $L$ and $J(L_{\st} + \varepsilon)$ are exactly the 
same, if $\angle A'C'D \ne \angle B'C'D$, then either $\overrightarrow{C'E}$ 
or a perturbation in the direction of $\overrightarrow{EC'}$ will cause both 
$|A'C'| + |B'C'| + |C'C|$ and $|A'C'| + |B'C'| + 2|C'C|$ to decrease, which 
means that $L$ and $J(L_{\st} + \varepsilon)$ can be simultaneously reduced. 
This contradicts the assumption that $L$ is the smallest budget for which the 
current objective $J(L_{\st} + \varepsilon)$ is possible. Since this cannot 
happen, it must be the case that $\angle A'C'D = \angle B'C'D$ in an optimal 
network configuration. That is, $CC'$ is a bisector of $\angle A'C'B'$. By 
symmetry, $BB'$ is a bisector of $\angle A'B'C'$ and $AA'$ is a bisector of 
$\angle B'A'C'$. 
\end{proof}

Before moving on to showing that $\triangle A'B'C'$ is equilateral, we note 
that Lemma~\ref{l:bisector} does not depend on $\varepsilon$ being small. 
Moreover, the result continues to hold if there are one or two anchors, which
can be readily verified. 

\begin{lemma}[Anchor Bisector]
\label{l:anchor-bisector}
For three terminals $A, B, C$, suppose $C'$ is an internal anchor connected to 
$C$ in an optimal network structure $N(L)$. Then $CC'$ bisects the angle formed
by the other two outgoing edges from $C'$. 
\end{lemma}

We now prove a key structural property of \bsn for three terminals involving 
three anchors.  

\begin{theorem}[Steiner Triangle for Three Anchors]
\label{t:equilateral}
For three terminals $A, B$, and $C$ with a Steiner point, assume 
that $N(L_{\st} + \varepsilon)$ is composed of the Steiner triangle 
$\triangle A'B'C'$ and segments $AA'$, $BB'$ and $CC'$. Then $\triangle A'B'C'$
is an equilateral triangle with its center being the Steiner point of the 
terminals. The center of $\triangle A'B'C'$ is the intersection point of $AA', 
BB'$ and $CC'$.
\end{theorem}
\begin{proof}
Again assuming an optimal solution, extend line segments $AA'$, $BB'$, and 
$CC'$ so that they intersect (see Fig.~\ref{fig:bisector-2}). 

\begin{figure}[h]
\begin{center}
\begin{overpic}[width=2.3in,tics=5]
{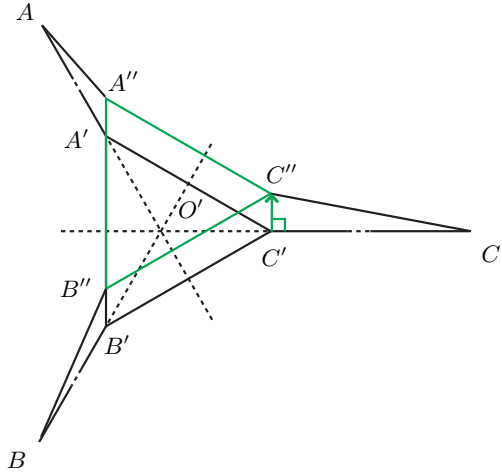}
\put(-4,97){{\small $A$}}
\put(-6,-5){{\small $B$}}
\put(102,42){{\small $C$}}
\put(7,68){{\small $A'$}}
\put(16,21){{\small $B'$}}
\put(52,41){{\small $C'$}}
\put(33,52){{\small $O'$}}
\put(17,81){{\small $A''$}}
\put(6,34){{\small $B''$}}
\put(53,60){{\small $C''$}}
\end{overpic}
\end{center}
\vspace*{-2mm}
\caption{Applying a perturbation to $\triangle A'B'C'$ that lifts it 
vertically along $CC'$, which keeps the length of $CC'$ unchanged in 
a first order approximation.}
\label{fig:bisector-2}
\end{figure} 

Because they are bisectors of $\triangle A'B'C'$, by Lemma~\ref{l:bisector}, 
they must meet at the same point $O'$. 
For this setting, we again apply a perturbation argument used in proving 
Lemma~\ref{l:bisector}, this time lift the entire $\triangle A'B'C'$ in a 
direction perpendicular to $CC'$. Let the perturbed triangle be 
$\triangle A''B''C''$. Using the same argument, this time applied to the 
length changes of $AA'$ and $BB'$, we can reach the conclusion that the line 
$CC'$ must be a bisector of $\angle  AO'B$. In other words, shifting $AA'$ and 
$BB'$ synchronously will not reduce the objective function only if $CC'$ bisects 
$\angle AO'B$.

Similarly, $AA'$ must be a bisector of 
$BO'C$ and $BB'$ must be a bisector of $AO'C$. Using that $CC'$ bisects 
$AO'B$ and $A'C'B'$, it can be derived that $\angle O'A'C' = \angle O'B'C'$, 
which in turn shows that $\angle B'A'C' = \angle A'B'C'$. By symmetry, it can 
then be concluded that $\triangle A'B'C'$ is an equilateral triangle. 
This further shows that $\angle A'O'B' =\angle A'O'C' = \angle B'O'C' = 2\pi/3$,
implying that $O'$, the center of $\triangle A'B'C'$, is the Steiner point $O$
of the terminals. 
\end{proof}

From Theorem~\ref{t:equilateral}, we can draw the following conclusion. For 
three terminals with a Steiner point, as the budget $L$ goes just 
beyond $L_{\st}$, an equilateral triangle will ``grow'' out the Steiner point
toward the terminals. Moreover, whenever there are three anchors, they 
must form an equilateral triangle. All such equilateral triangles have their 
vertices lying on the line segments formed by the terminals and the Steiner 
point, as illustrated in Fig.~\ref{fig:three-anchors}. We have not yet show, 
however, that as $L$ grows, the anchors cannot go from three to fewer and then 
become three again. We delay this after the structures with fewer anchors are 
characterized.

\begin{figure}[h]
\begin{center}
\begin{overpic}[width=2.59in,tics=5]
{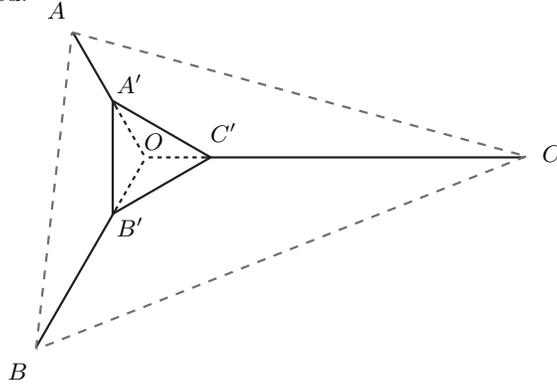}
\put(3,67){{\small $A$}}
\put(-5,-6){{\small $B$}}
\put(103,38){{\small $C$}}
\put(17,52){{\small $A'$}}
\put(17,23){{\small $B'$}}
\put(36,42){{\small $C'$}}
\put(22.5,40.5){{\small $O$}}
\end{overpic}
\end{center}
\caption{For three terminals with a Steiner point (which is always internal), 
when there are three anchors, they always form an equilateral triangle.}
\label{fig:three-anchors}
\end{figure} 

\subsection{One and Two Anchors}
If there are two anchors, they must both be connected to one shared terminal,
say $A$, and each connecting to a unique terminal in $B$ and $C$. Let the 
anchors be $B'$ and $C'$. $N(L)$ then consists of five segments $AB'$, $AC'$, 
$BB'$, $CC'$, and $B'C'$. It can be shown that $\triangle AB'C'$ is an 
isosceles triangle (see, e.g., Fig.~\ref{fig:two-anchors}).

\begin{figure}[h]
\begin{center}
\vspace{1mm}
\begin{overpic}[width=2.59in,tics=5]
{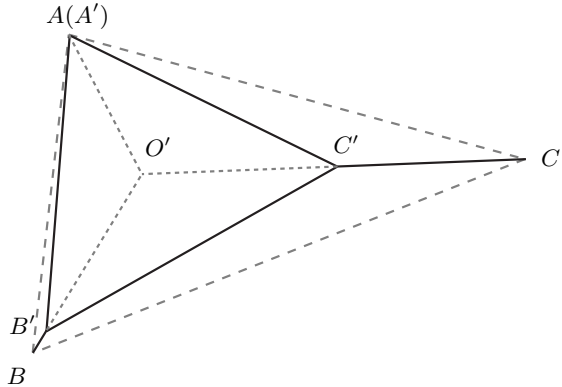}
\put(3,67){{\small $A$($A'$)}}
\put(-5,-6){{\small $B$}}
\put(103,38){{\small $C$}}
\put(-4.5,4){{\small $B'$}}
\put(61,40.5){{\small $C'$}}
\put(23,40){{\small $O'$}}
\end{overpic}
\end{center}
\caption{For three terminals with a Steiner point, when there are two 
anchors, they always form an isosceles triangle with one of the terminals.}
\label{fig:two-anchors}
\end{figure}

\begin{proposition}[Steiner Triangle for Two Anchors]\label{l:s2a}
For three terminals $A, B$, and $C$ with a Steiner point, if the 
optimal network $N(L)$ has two anchors $B', C'$, then these two anchors form 
an isosceles triangle with one of the terminals, e.g., $A$. $AB' = AC'$. 
\end{proposition}
\begin{proof}
By Lemma~\ref{l:anchor-bisector}, $BB'$ bisects $\angle AB'C'$ and $CC'$ 
bisects $\angle AC'B'$. Let the extensions of $BB'$ and $CC'$ meet at $O'$
(see Fig.~\ref{fig:two-anchors}). Then $AO'$ bisects $\angle B'AC'$. 
Using the perturbation argument from the proof of Theorem~\ref{t:equilateral}, 
applied to perturb the lengths of $BB'$ and $CC'$, we can show that $AO'$ is 
also a bisector of $\angle B'O'C'$ (we do this by ``rotating'' $\triangle 
AB'C'$ with center $A$ slightly). 
This means that $\angle B'O'A = \angle C'O'A$, which in turn implies that 
$\angle AB'O' = \angle AC'O'$ and further implies $\angle AB'C' = \angle 
AC'B'$. Therefore, $\triangle AB'C'$ is an isosceles triangle and $AB' = AC'$.  
\end{proof}

Following the same line of reasoning, when there is a single anchor in an 
optimal network $N(L)$, e.g., $C'$ that is connected to $A, B$, and $C$, if 
$C'$ is not the Steiner point, $N(L)$ must contain one of $AB$, $BC$, and $AC$.
Suppose $N(L)$ contains $AB$, then all we know is that $CC'$ must bisect 
$\angle AC'B$. See Fig.~\ref{fig:spectrum}(d) for an example.

\section{Evolution of the Budgeted Steiner Network}\label{sec:internal}
\subsection{With Steiner Point}
Having established the optimal configuration when there are $1$-$3$ anchors, 
we now piece them together to understand the evolution of the network. 
Intuitively, as the budget $L$ increases, the evolution of the optimal 
network $N(L)$ would look like that shown in Fig.~\ref{fig:spectrum}, going 
from Steiner tree to having three anchors, then two, then one, and finally 
becoming the triangle of the three terminals. To show this is the actual network
evolution pathway, however, we must show that there cannot be discrete 
jumps in \bsn structures, e.g., going from three anchors to two anchors and then 
back to three anchors.

We proceed to show that the sequence in Fig.~\ref{fig:spectrum} is indeed how 
$N(L)$ evolves as $L$ increases by analyzing how $J(L)$ changes as $L$ 
changes, i.e., $\frac{dJ}{dL}$.

\begin{lemma}[Rate of Change at Anchors]
For three terminals $A$, $B$, and $C$, let $C'$ be an anchor connected to 
$C$. Let the angle formed by the other two edges emanating from $C'$ other than
$CC'$ be $2\alpha$. As $C'$ moves closer to $C$, the rate of change to the 
objective function $\frac{dJ}{dL}$ due to the change to $CC'$ is 
\begin{align}\label{eq:rate-of-change}
\frac{dJ}{dL} = \frac{2\cos\alpha - 2}{2\cos\alpha - 1}. 
\end{align}
\end{lemma}
\begin{proof}
Fig.~\ref{fig:roc} shows the setting where $C'$ is moved along $C'C$ for a 
small amount. By the bisector Lemma~\ref{l:anchor-bisector}, the addition of 
length (in green) to the two edges coming out of $C'$ that are not $CC'$ is $2|EC'|$ 
while the reduction of length to $|CC'|$ is $|C'E|/\cos \alpha$ (the red 
segment). Therefore, the change to the budget due to this is $\Delta L = 2 
|C'E| - |C'E|/\cos \alpha$. 

\begin{figure}[h]
\begin{center}
\vspace{1mm}
\begin{overpic}[width=1.8in,tics=5]
{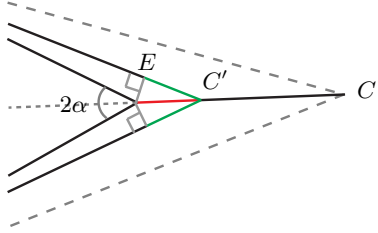}
\put(103,38){{\small $C$}}
\put(58,40.5){{\small $C'$}}
\put(39,46){{\small $E$}}
\put(16.5,34){{\small $2\alpha$}}
\end{overpic}
\end{center}
\caption{Moving $C'$ along $C'C$ for a small amount.}
\label{fig:roc}
\end{figure} 

On the other hand, the change to the objective function value is $\Delta J
= -(2|C'E|/\cos \alpha - 2|C'E|)$ because $C'C$ contributes to two shortest
paths. Dividing $\Delta J$ over $\Delta L$ yields Eq.~\ref{eq:rate-of-change}. 
\end{proof}

\begin{corollary}[Range of Change, Three Anchors]
For three terminals, when there are three anchors, 
\begin{align}\label{eq:rate-of-change-3c}
\frac{dJ}{dL} = \frac{1-\sqrt{3}}{2}.
\end{align}
\end{corollary}
\begin{proof}
For three anchors, $\alpha$ in Eq.~\eqref{eq:rate-of-change} is $\pi/6$. We then 
have $dJ/dL = (\sqrt{3} - 2)/(\sqrt{3}-1) = (1 - \sqrt{3})/2$.
\end{proof}

\begin{corollary}[Range of Change, One and Two Anchors]
For three terminals, when there are one of two anchors, let the angle formed
at the anchor belonging to the triangle structure of the network be $2\alpha$, 
then, 
\begin{align}\label{eq:rate-of-change-3c}
\frac{dJ}{dL} = \frac{2\cos\alpha - 2}{2\cos\alpha  - 1}. 
\end{align}
\end{corollary}

Since $0 < 2\alpha \le \pi/2$, $\alpha \in (0, \pi/4]$. Let $\cos\alpha = x$, 
$x \in [\frac{\sqrt{2}}{2}, 1)$. Eq.~\ref{eq:rate-of-change} becomes $g(x)
 = \frac{2x - 2}{2x - 1}$. It is straightforward to derive (using derivatives)
that $g(x)$ is negative on the given range of $x$ and monotonically increases 
to $0$ as $x \to 1$. This means, with reference to Fig.~\ref{fig:roc}, that 
the magnitude of $\frac{dJ}{dL}$ becomes smaller as $C'$ gets closer to $C$ 
($\alpha$ decreases). This allows us to show that $J(L)$ decreases faster when 
there are more anchors. We begin with showing that internal angles at anchors
cannot exceed $\pi/3$. 

\begin{lemma}[Feasible Anchor Angle Configurations]\label{l:feasible}
For three terminals and an optimal Steiner network, the internal angles of the 
triangular structure of the network at non-terminal anchors are always no more than $\pi/3$. 
\end{lemma}
\begin{proof}
For three anchors, we have shown they must assume an equilateral triangle 
configuration. Suppose that in a two-anchor network configuration, the optimal 
network has internal angles at non-terminals anchors larger than $\pi/3$. For example, 
suppose that in Fig.~\ref{fig:two-anchors},
$\angle AB'C' = \angle AC'B' > \pi/3$. This requires that $\angle B'AC' < 
\pi/3$. Now, suppose we push down the triangle $A'B'C'$ along $AA'$ by a small
$\delta > 0$ and retract along $B'B$ and $C'C$ so that $L$ remains unchanged. 
Because $0 > \frac{dJ}{dL}|_{A'} > \frac{dJ}{dL}|_{B'} = \frac{dJ}{dL}|_{C'}$, 
this means that $J$ will actually decrease due to the change. Therefore, the 
configuration cannot be optimal. 

The same argument also applies to the single anchor case: if the internal angle
at the single anchor is larger than $\pi/3$, the at least one of the two other 
internal angles must be smaller than $\pi/3$. 
\end{proof}

We are now ready to establish the evolution pathway of the optimal Steiner 
network for three terminals with Steiner points. 

\begin{theorem}[Network Evolution, with Steiner Point]
For three terminals $A, B$, and $C$ with a Steiner point $O$, as the 
budget $L > L_{\st}$ increases, the optimal Steiner network $N(L)$ will first 
grow an equilateral triangle, $\triangle A'B'C'$, out of $O$ toward the three 
terminals. The internal angles of $\triangle A'B'C'$ are bisected by $AA', BB'$
and $CC'$. 
The growth continues until one of the anchors, say $A'$, reaches terminal $A$, corresponding to the largest internal angle of $\triangle ABC$. Then, an 
isosceles triangle continuous to grow in place of the equilateral triangle, 
with its two internal angles $\angle AB'C'$ and $AC'B'$ bisected by $BB'$ and 
$CC'$, respectively , until one of the two anchors $B'$ reaches 
a second terminal, say $B$, that corresponds to the second largest angle 
of $\triangle ABC$. Finally, the network grows as $C'$ finally reaches $C$, 
with $CC'$ always bisecting $\angle AC'B$. 
\end{theorem}
\begin{proof}
Without loss of generality, assume that $\angle BAC \ge \angle ABC \ge \angle ACB$. 
By Lemma~\ref{l:anchor-number-3-isp} and Theorem~\ref{t:equilateral}, the 
initial optimal network when $L = L_{\st} + \varepsilon$ has an equilateral 
triangle $A'B'C'$ growing out of the Steiner point $O$, with $AA'$, $BB'$, 
and $CC'$ bisecting $\angle B'A'C'$, $\angle A'B'C'$, and $\angle A'C'B'$, 
respectively. 
By Lemma~\ref{l:feasible}, before $\triangle A'B'C'$ reaches $A$ as an 
equilateral triangle ($AA'$ is shorter than than $BB'$ and $CC'$ when 
$\angle BAC$ is the largest angle of $\triangle ABC$), it cannot happen that 
the optimal network jumps to a configuration where one anchor disappears. To 
see that this is the case, suppose the network jumps to a configuration where 
$A'$ merges with $A$. This would force $\triangle A'B'C'$ to have 
$\angle B'A'C' < \pi/ 3 < \angle A'B'C' = \angle A'C'B'$, which is not 
possible. The situation gets worse if $B'$ merges with $B$ or $C'$ merges with 
$C$. 
Using a similar argument, we can show that it is also not possible for 
the optimal network to jump from three anchors to having a single anchor 
without the equilateral $\triangle A'B'C'$ reaching its maximum girth. 
Using the same approach, we can also show that it is not possible to ``jump''
from a two-anchor configuration to a single anchor configuration without 
the anchor $B'$ reaching $B$, as the isosceles triangle expands. 
\end{proof}

\subsection{No Steiner Point}
When an angle of $\triangle ABC$, say $\angle BAC$, is larger than $2\pi/3$, 
$A$ acts as a ``Steiner'' point. In this case, it becomes impossible for the 
optimal network $N(L)$ to have three internal anchors. 

\begin{lemma}[Anchor Multiplicity]\label{l:am}
For three terminals without a Steiner point, the optimal network 
$N(L)$ for any $L$ cannot have three anchors. 
\end{lemma}
\begin{proof}
If there are three anchors, Theorem~\ref{t:equilateral} must hold. However, 
this is impossible if one of the angles formed by the terminals is equal to or 
larger than $2\pi/3$. Referring to Fig.~\ref{fig:three-anchors}, suppose that
$\angle BAC \ge 2\pi/3$. However, also by Theorem~\ref{t:equilateral}, 
$\angle BOC = 2\pi/3$, which is not possible. 
\end{proof}

Following similar reasoning used for establishing the case where the Steiner 
point is in the interior of $\triangle ABC$, the evolution of the optimal network 
for the current setting goes through the following phases (assuming terminals
 $A$, $B$, and $C$, and $\angle BAC \ge 2\pi/3$): 
\begin{enumerate}
\item The budget $L$ is sufficient to cover the shortest edge of 
$\triangle ABC$ but less than $L_{\st}$. In this case, $N(L)$ contains one edge
of $\triangle ABC$
\item The budget $L$ equal to $L_{\st}$. In this case, $N(L)$ is the Steiner 
tree comprised of $AB$ and $AC$.
\item For $L = L_{\st} + \varepsilon$ for small positive $\varepsilon$, a small 
isosceles triangle grows out from $A$, producing a configuration as shown in 
Fig.~\ref{fig:spectrum-2}(a). The network satisfies the bisector requirement 
given by Lemma~\ref{l:s2a}. As $L$ increases, the isosceles triangle expands 
with the bisector structure in place, until one of the vertex of the triangle
hits a terminal ($B$). 
\item As one of the two anchors merge with a terminal, the other anchor will 
continue to march toward the last terminal ($C$) as $L$ increases, eventually 
merge with that terminal. A snapshot of this process is given in 
Fig.~\ref{fig:spectrum-2}(b). 
\end{enumerate}

\begin{figure}[h]
\begin{center}
\begin{overpic}[width=3.8in,tics=5]
{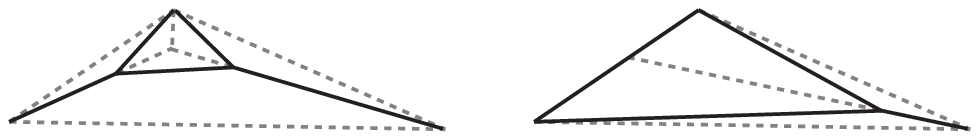}
\put(0, 0){{\small $B$}}
\put(18,14){{\small $A$}}
\put(47,0 ){{\small $C$}}
\put(22,-4){{\small (a)}}
\put(72,-4){{\small (b)}}
\end{overpic}
\end{center}
\caption{A spectrum of optimal Euclidean \bsn network structures (solid lines) 
for three terminals in a typical setup where $\angle BAC \ge 2\pi/3$, as the allowed budget increases.}
\label{fig:spectrum-2}
\end{figure} 


\section{Conclusion and Discussions}\label{sec:conclusion}
In this work, we propose the \emph{budgeted Steiner network} (\bsn) 
problem to study shortest path structures among multiple 
terminals under a path length budget.
We establish the precise evolution of the \bsn structure for 
three arbitrarily located terminals where paths between 
each pair of terminals have equal importance. 
It is clear that the characterization yields efficient 
algorithms for computing optimal \bsn structures for 
any given $3$-terminal setup and length budget. 

We mention that, beside potential real-world applications, 
\bsn structures also appear in natural processes. For example,
when a cell dies among a groups of cells in an organism, 
the perishing cell gradually gets absorbed by surrounding 
tissues. The process closely mimics the inverse of the emergence
of the Steiner triangle for three anchors, known as the 
$T_2$ process \cite{weaire1984soap}. 

The current work just begins to scratch the surface of the study 
of \bsn; we mention a few interesting directions for future study:
(1) With the equal path weight case solved for three terminals, 
it seems possible to extend the analytical techniques developed in 
this study to work for the case where some paths connecting the 
terminals are more important than others; 
(2) It would also be interesting to characterize \bsn 
structures for four or more terminals, in which the weight 
constraints will cause even the starting structure to differ from 
Steiner trees for the same number of terminals;
(3) It is interesting to explore how \bsn structures are affected 
by obstacles that fall in the convex hull of the terminals. In this 
case, discontinuities in the evolution of the network can be 
unavoidable; and 
(4) As an alternative to analytical approaches, it is interesting 
to explore approximation algorithms as well as numerical methods 
for computing optimal \bsn structures. Numerical methods appear 
promising when the number of terminals are limited, which means that
the number of anchors and the possible \bsn structures are also 
limited, allowing enumeration in searching for 
the optimal one. 

%

\bibliographystyle{plain}
\bibliography{./bib/arc,./bib/others}

\end{document}